\documentclass[11pt,letterpaper]{article}
\usepackage[margin=1in]{geometry}
\usepackage{fullpage}
\usepackage{float}
\usepackage{amsfonts,amsmath,amssymb,amsthm,graphicx}

\usepackage{amstext}
\usepackage{url}
\usepackage{subfig}
\usepackage{multirow}
\usepackage{caption}
\usepackage{comment}
\usepackage{hyperref}
\usepackage[capitalize]{cleveref}

\usepackage[authoryear,round]{natbib}
\usepackage{enumerate}

\allowdisplaybreaks[1]

\usepackage[linesnumbered,ruled]{algorithm2e}
\SetKwComment{Comment}{/* }{ */}

\newtheorem{theorem}{Theorem}[section]

\newtheorem{lemma}[theorem]{Lemma}
\newtheorem{proposition}[theorem]{Proposition}
\newtheorem{claim}[theorem]{Claim}

\newtheorem{assumption}[theorem]{Assumption}

\newtheorem{example}[theorem]{Example}

\newcommand{\given}{\,|\,}

\newcommand{\prob}[2][]{\text{\bf Pr}\ifthenelse{\not\equal{}{#1}}{_{#1}}{}\!\left[{\def\givenn{\middle|}#2}\right]}
\newcommand{\expect}[2][]{\text{\bf E}\ifthenelse{\not\equal{}{#1}}{_{#1}}{}\!\left[{\def\givenn{\middle|}#2}\right]}

\newcommand{\tparen}{\big}
\newcommand{\tprob}[2][]{\text{\bf Pr}\ifthenelse{\not\equal{}{#1}}{_{#1}}{}\tparen[{\def\given{\tparen|}#2}\tparen]}
\newcommand{\texpect}[2][]{\text{\bf E}\ifthenelse{\not\equal{}{#1}}{_{#1}}{}\tparen[{\def\given{\tparen|}#2}\tparen]}

\newcommand{\sprob}[2][]{\text{\bf Pr}\ifthenelse{\not\equal{}{#1}}{_{#1}}{}[#2]}
\newcommand{\sexpect}[2][]{\text{\bf E}\ifthenelse{\not\equal{}{#1}}{_{#1}}{}[#2]}

\newcommand{\eps}{\epsilon}

\newcommand{\rev}{{\rm Rev}}

\newcommand{\opt}{{\rm OPT}}
\newcommand{\optval}{\rm{OPT\text{-}EF}}

\newcommand{\optefs}{{\rm OPT\text{-} EFS}}

\usepackage{xcolor}

\usepackage{array}
\usepackage{booktabs}

\usepackage{color-edits}
\addauthor{yl}{red}

\newcommand{\NP}{\mbox{\sf NP}}

\begin{document}

\title{Strict Fairness at What Cost? Envy-Free Contracts with Subsidies}
\author{
Matteo Castiglioni\thanks{DEIB, Politecnico di Milano. Email: \texttt{matteo.castiglioni@polimi.it}}
\and Junjie Chen\thanks{Department of Economics, National University of Singapore. Email: \texttt{junjchen9-c@my.cityu.edu.hk}}
\and Yingkai Li\thanks{Department of Economics, National University of Singapore. Email: \texttt{yk.li@nus.edu.sg} }
}
\date{}

\begin{titlepage}
	\clearpage\maketitle
	\thispagestyle{empty}

\begin{abstract}

We study algorithmic fair contract design, where a principal designs task-level contracts and fairly delegates a set of tasks to a set of agents. 
Prior work reveals a fairness-revenue dilemma: exact envy-free (EF) contracts may have an unbounded price of fairness (PoF), while approximate notions avoid this unboundedness only by weakening strict fairness.
To address this dilemma, we propose a novel scheme, called {\it Envy-free Contracts with Subsidies} (EFS), in which the principal may additionally offer agent-specific subsidies to restore strict fairness.
Our main technical result is a tight characterization of the price of fairness for EFS contracts. In sharp contrast to EF contracts, whose PoF can be unbounded, we show that the PoF of EFS contracts is $n^{n+O(1)}$, where $n$ is the number of agents.
Moreover, EFS contracts can outperform EF contracts by an arbitrarily large factor in terms of the principal's revenue.
Finally, we present the complexity landscape: computing optimal EFS contracts is NP-hard in general, whereas a polynomial-time algorithm exists when the number of tasks is constant.

\end{abstract}

\end{titlepage}

\section{Introduction}

Contract design traditionally focuses on incentive provision and revenue maximization. In multi-agent settings, however, fairness among agents may also affect their willingness to exert effort~\citep{akerlof1990fair}. This concern is especially salient in delegated-work environments, where outcomes are observable but effort is not, and where heterogeneous agents may face different costs, opportunities, or task difficulties. As a result, a revenue-maximizing contract may induce task assignments and payments that some agents perceive as unfair.

Motivated by this issue, \citet{castiglioni2025faircontracts} recently initiated the algorithmic study of fair contracts in a multi-agent, multi-task delegation model. In their framework, a principal allocates tasks and designs task-level contracts while accounting for both performance and fairness. The central fairness notion is envy-freeness (EF): each agent must weakly prefer its own assigned tasks and associated contracts to those assigned to any other agent. 

Prior work on fair contract design reveals a fundamental tension between fairness and the principal's revenue. Exact envy-freeness provides a strong fairness guarantee, but it can be prohibitively costly for the principal, potentially causing an arbitrarily large revenue loss.
\citet{castiglioni2025faircontracts} show that this issue can be mitigated by relaxing EF to approximate notions such as EF1 (envy-free up to one item) and $\eps$-EF. 
These relaxations are meaningful because they control the extent of unfairness: envy is permitted only after removing a task from the envied bundle, or only up to an $\eps$ amount. Nevertheless, they weaken the fairness guarantee itself. Some envy may remain, even if only to a limited degree, and therefore these notions do not provide exact envy-freeness.

Such a tension leads to a fundamental question in fair contract design:
\begin{center} {\it Can we restore strict envy-freeness among agents while avoiding the unbounded revenue loss to the principal?}\end{center}
This paper gives an affirmative answer to this question. 
Instead of weakening the fairness requirement, the principal may supplement task-level contracts with compensatory transfers (i.e., subsidy). 
Conceptually, the task-level contract remains responsible for inducing effort, while the subsidies correct distributional imbalances among agents. 
{At first glance, allowing subsidies may seem to trivialize fairness: principals can compensate the disadvantaged agents and easily find an envy-free allocation as~\citep{halpern2019fair}. However, subsidies are costly to the principal and interact with revenue maximization, incentive constraints, and task allocation. Thus, the key question is not whether envy can be eliminated, but how costly it is and whether the optimal contracts can be computed efficiently. Moreover, our problem is not equivalent to subsidy minimization as ~\citep{halpern2019fair}, since an EF contract (with {\it zero} subsidies) can be easily found  but may cause a large revenue loss.}

Subsidies provide a separate instrument for restoring exact fairness, and this separation is especially natural in platform work and crowdwork, where contracts are posted at the task level and compensatory transfers are used to correct unequal realized pay. 
In platform work, earnings guarantees and top-ups, such as Uber's Proposition 22 materials~\citep{uberprop22} and Lyft's earnings commitment~\citep{lyft2024earnings}, illustrate how a platform can preserve assignment and incentive mechanisms while using compensatory transfers to address unequal earnings.
Crowdwork provides an even closer task-level example: requesters post base payments, while bonus mechanisms can correct low realized hourly wages. The Fair Work system of \citet{whiting2019fairwork}, for instance, automatically pays bonuses so that workers reach a target hourly wage, against a backdrop of documented low and uneven realized hourly wages on Amazon Mechanical Turk~\citep{hara2018datadriven}.


Motivated by this perspective, we propose {\it Envy-free Contracts with Subsidies} (EFS), a fair contract scheme in which the principal chooses task assignments, task-level contracts, and individually tailored nonnegative subsidies. Our contributions are summarized as follows:

\begin{itemize}
    \item We show that the optimal principal's revenue of our EFS contracts can outperform that of EF contracts by an arbitrarily large factor. Hence, subsidies are not cosmetic, but they fundamentally improve the revenue-fairness tradeoff. This separation comes from instances where exact EF forces the principal to recruit less-effective agents, while subsidies compensate envy and allow the principal to contract with productive agents.
\item Our main technical contribution is a tight characterization of the price of restoring strict envy-freeness with subsidies. In contrast to EF contracts, whose price of fairness can be unbounded, EFS contracts preserve strict envy-freeness among agents while admitting a bounded worst-case price of fairness. Specifically, we prove an upper bound of $O((n+1)^n)$ and complement it with a lower-bound construction of $\Omega(n^{n-1})$. 
These bounds differ only by a polynomial factor, and therefore yield an $n^{n+O(1)}$ characterization of the price of fairness for EFS contracts.
A key ingredient in our analysis is a characterization of optimal single-task EFS contracts. This characterization allows us to reduce the main price-of-fairness analysis to the single-task setting. For the upper bound, we first prove the guarantee for a single task and then extend it to multiple tasks using the additivity of the principal's revenue. For the lower bound, the characterization guides a matching-in-the-exponent construction whose instance consists of a single task.
    \item We show that computing an optimal EFS contract in general is \NP-hard. This is proved via a non-trivial reduction from the $3$-Partition problem. Nevertheless, when the number of tasks is a constant, we present a polynomial-time algorithm to compute the optimal EFS contract. Moreover, we establish an algorithmic connection to EF contracts by showing that one can reduce any EFS contract instance to an EF contract instance in polynomial time. 
\end{itemize}

\subsection{Related Work}
Our work is related to the literature on algorithmic contracts and fair division.

\paragraph{Related work on algorithmic fair contracts.} The closest work to ours is by \citet{castiglioni2025faircontracts}, who recently initiated the algorithmic study of fair contract design and introduced EF, EF1, and $\epsilon$-EF contracts in a multi-agent, multi-task model. Their results reveal a central tension: exact envy-freeness can be prohibitively costly, while approximate notions improve revenue by weakening the fairness guarantee. Our work takes this tension as the starting point and studies subsidies as a compensatory instrument for restoring exact envy-freeness.
Subsequent work studies fair contracts in a single-task, multi-agent collaborating setting, where envy-freeness and non-discrimination are employed as fairness constraints.~\citep{castiglioni2025fair,feldman2026equal,ding2026multi}.

\paragraph{Related work on algorithmic contracts.} 
Contract design has recently received increasing attention in the computer science community. To the best of our knowledge, the algorithmic study of contract design dates back to~\citet{babaioff2006combinatorial}. A growing line of work has investigated Bayesian contract design under adverse selection and moral hazard. \citet{alon2021contracts,alon2022bayesian} investigate settings with single-dimensional agent types, whereas \cite{castiglioni2021bayesian,castiglioni2022designing,guruganesh2021contracts} present the complexity of computing Bayesian contract design with multi-dimensional agent types. More recently, \citet{castiglioni2025reduction} establish a fundamental algorithmic equivalence between single-dimensional and multi-dimensional Bayesian contract design. We refer readers to the recent survey by~\citet{dutting2024algorithmic} and ~\citet{feldman2025combinatorial}.

Our work is also closely related to the recent literature on combinatorial contract design~\citep{dutting2024combinatorial,dutting2023multi,dutting2025multi}. The closest one to ours is by \citet{alon2025multiprojectcontracts}, who study assigning multiple agents to a set of tasks, where multiple agents may collaborate on a single task. In contrast, our model assigns multiple tasks to agents, where one task is assigned to at most one agent, and each agent can work on multiple tasks.

\paragraph{Related work on fair division.} The paper relates to the literature on fair division of goods and chores. One typical fairness notion employed is envy-freeness, which has long been a central and extensively studied fairness notion since~\citep{foley1967resource}. Several prominent relaxations and extensions of envy-freeness have been proposed, including envy-free up to one item (EF1)~\citep{lipton2004approximately,budish2011combinatorial} and envy-freeness up to any
item (EFX)~\citep{caragiannis2019unreasonable}. A closely related recent line of work studies fair division with subsidies~\citep[e.g.,][]{halpern2019fair,elmalem2025whoever,wu2025little,caragiannis2021computing}. {These works use monetary transfers to eliminate or reduce envy in allocation problems. In contrast, our setting combines subsidies with moral-hazard contract design: subsidies interact with task allocations, agents' incentives and the principal's revenue maximization. Thus, our main question is not only whether envy can be eliminated, but also the revenue cost of restoring strict envy-freeness through subsidies.}

\section{Model}

We consider a contract design model with moral hazard. Specifically, there is a finite collection of tasks $\mathcal M$, and the principal must delegate the tasks to a finite set of agents $\mathcal N$. Let $|\mathcal M|=m$ and $|\mathcal N|=n$. 
In our model, the agents and the tasks are heterogeneous: if agent $i\in\mathcal N$ exerts effort on task $j$, the task succeeds with probability $p_{i,j}\ge 0$, and agent $i$ incurs an effort cost of $c_{i,j}\ge 0$.
Alternatively, if the agent chooses to shirk, the task succeeds with a probability of zero, and the agent incurs no cost. 
Finally, if task $j\in\mathcal M$ is successful, it generates a reward $r_j\ge 0$ for the principal. Without loss of generality, we assume the reward $r_j \in[0, 1]$ for every task $j\in \mathcal{M}$.

The principal observes task outcomes but not the agents' effort
choices. To induce effort, the principal specifies outcome-contingent
payments. Following \citep{castiglioni2025faircontracts}, we restrict attention to task-level linear contracts. For each task $j$, the principal chooses a share $\alpha_j\in[0,1]$; if task
$j$ succeeds, the agent assigned to this task receives payment $\alpha_j r_j$, while no payment is made upon failure. Thus, when agent $i$ is assigned task $j$ and exerts effort, its expected utility
from this task is $\alpha_j p_{i,j}r_j-c_{i,j}$.

A task allocation is denoted by $S=(S_1,\ldots,S_n)$, where $S_i\subseteq\mathcal M$ is the set of
tasks assigned to agent $i$. We impose that tasks are not duplicated: $S_i\cap S_{i'}=\emptyset$ for all $i\neq i'$. 
In addition to choosing the allocation $S$ and the task-level shares
$\alpha$, the principal may provide each agent with a nonnegative
subsidy. Let $s_i\ge 0$ denote the subsidy given to agent $i$.
A contract with subsidies is a triple $(S,\alpha,s)$, where
$s=(s_1,\ldots,s_n)$.

 Whenever task $j$ is assigned to agent $i$ by the principal,
the contract must make effort weakly optimal for that agent. Formally, we impose the following effort constraints,
\begin{equation}
\tag{Effort constraints}
\label{effort-constraints}
    \alpha_jp_{i,j}r_j-c_{i,j}\ge 0,
    \quad \forall i\in\mathcal N,\ \forall j\in S_i .
\end{equation}
These constraints ensure that every assigned agent is willing to exert effort on every assigned task.

Given any bundle of tasks, an agent evaluates that bundle using its
own success probabilities and costs. Recall that an agent can always shirk on
a task that would yield negative utility. By~\eqref{effort-constraints}, agent $i$'s utility from its own
assigned bundle $S_i$ and subsidy $s_i$ is
    $\sum_{k\in S_i}(\alpha_kp_{i,k}r_k-c_{i,k})+s_i$, while its utility from other agent $j$'s assigned bundle and subsidy is $\sum_{k\in S_j}
    \max\{\alpha_kp_{i,k}r_k-c_{i,k},0\}+s_j$. 

We say that $(S,\alpha,s)$ is an {\it Envy-free Contract with Subsidies}
(EFS contract) if no agent prefers another agent's bundle and subsidy
to its own. That is, for every ordered pair of distinct agents
$i,j\in\mathcal N$,
\begin{equation}
\tag{EFS constraints}
\label{efsconstraints}
    \sum_{k\in S_i}(\alpha_kp_{i,k}r_k-c_{i,k})+s_i
    \ge
    \sum_{k\in S_j}
    \max\{\alpha_kp_{i,k}r_k-c_{i,k},0\}+s_j .
\end{equation}
The subsidy terms enter the fairness comparison directly:
an agent compares not only the assigned task-contract bundle, but also
the monetary compensation attached to that bundle.
Note that the standard envy-free (EF) contract model is recovered as the special case
where subsidies 
$s_i=0$ for all $i\in\mathcal N$, which is formulated as 
\begin{equation}
\tag{EF constraints}
\label{efconstraints}
    \sum_{k\in S_i}(\alpha_kp_{i,k}r_k-c_{i,k})
    \ge
    \sum_{k\in S_j}
    \max\{\alpha_kp_{i,k}r_k-c_{i,k},0\},
     \forall i\neq j .
\end{equation}

Following \citep{castiglioni2025faircontracts}, we focus
on full-allocation contracts, meaning that every task must be assigned to some agent. 
An EFS contract $(S,\alpha, s)$ is full-allocation if
\begin{equation}\label{allocation-constraints}\tag{Allocation constraints}
    \bigcup_{i\in\mathcal N}S_i=\mathcal M,~
    S_i\cap S_{i'}=\emptyset, \quad\forall i\neq i',
\end{equation}
Following \citep{castiglioni2025faircontracts}, we maintain the following mild feasibility assumption.
\begin{assumption}
\label[assumption]{asp:usefulness}
For every task $j\in\mathcal M$, there exists an agent $i\in\mathcal N$
such that $p_{i,j}r_j-c_{i,j}\ge 0 .$
\end{assumption}
\cref{asp:usefulness} says that every task is useful for at least one agent
under the full-transfer contract $\alpha_j=1$. Under this assumption, a full-allocation EFS contract always exists.

\begin{proposition}\label[proposition]{propo_existend_efs}
Under \cref{asp:usefulness}, there exists a full-allocation EFS
contract, and one can be found in polynomial time.
\end{proposition}

\paragraph{The principal's optimization problem.}
For a feasible EFS contract $(S,\alpha,s)$, the principal's expected
revenue is
\[
    \rev(S,\alpha,s)
    =
    \sum_{i\in\mathcal N}\sum_{k\in S_i}
    (1-\alpha_k)p_{i,k}r_k
    -
    \sum_{i\in\mathcal N}s_i .
\]
The first term is the expected reward retained by the principal after
success-contingent payments, and the second term is the total subsidy
paid to the agents. The principal chooses the allocation, the linear
shares, and the subsidies to maximize this revenue subject to effort
feasibility, full allocation, and envy-freeness with subsidies:
\begin{align}
\tag{EFS-Program}
\label{efs_program}
    \max_{S,\alpha,s}\quad
    & \rev(S,\alpha,s) \nonumber\\
    \text{s.t.}\quad
    & \text{\eqref{effort-constraints}}, \text{\eqref{efsconstraints}} \nonumber\\
    & \text{\eqref{allocation-constraints}} \nonumber\\
    & s_i\ge 0,
      \qquad \forall i\in\mathcal N, \nonumber
\end{align}
We denote the optimal value of this program by $\optefs$. We also let
$\optval$ denote the optimal revenue under the standard EF constraints, namely the special case in which all subsidies are fixed to zero. Moreover, we define $\frac{0}{0} = 0$ and $\frac{c}{0}=+\infty$ for any $c>0$.

\section{EFS Contracts can be Arbitrarily Better than EF Contracts}

This section presents the advantages of EFS contracts. \citet{castiglioni2025faircontracts} shows that with exact EF fairness constraints, the principal needs to sacrifice an arbitrarily large amount of revenue to create envy-freeness among agents. They further show that relaxing exact EF constraints to EF1 and $\eps$-EF constraints can improve the revenue, but these relaxed notions indeed introduce a degree of {\it unfairness} among agents. In sharp contrast, our EFS contracts address both issues simultaneously: (i) EFS contracts can restore strict fairness, unlike EF1 and $\eps$-EF; (ii) It strictly improves the principal's revenue over EF contracts. Indeed, \cref{efs_over_ef} shows that EFS contracts can outperform EF contracts by an arbitrarily large factor in terms of the principal's revenue.

\begin{proposition}\label[proposition]{efs_over_ef}
    There exists an instance in which the ratio between the revenues of optimal EFS contracts and optimal EF contracts is arbitrarily large.
\end{proposition}

We prove  \cref{efs_over_ef} by the following example, which contains only two agents and a single task. One agent gives a high revenue while another agent is much less efficient.  Under EF contracts, the principal must assign the task to a less efficient agent and receives negligible revenue. However, under EFS contracts, a small subsidy to the less efficient agent allows the principal to assign the task to the efficient agent
and obtain much larger revenue.

\begin{example}\label{example52}
Consider the example in \citep{castiglioni2025faircontracts} parameterized by a constant  $\delta$$>0$. Suppose there is one task with reward $r=1$ and two agents.  Set the success probabilities and costs to $p_{1} = 10 \delta$, $c_1 =  \delta$ and $p_{2} = \frac{1}{2}$, $c_2 = \frac{1}{4}$. 
To achieve envy-freeness without subsidies, the principal must assign the task to agent $1$, and to maximize the revenue, the principal needs to offer a linear contract $\frac{c_1}{p_1}=\frac{1}{10}$ and gains a revenue of $(1-\frac{1}{10}) \cdot 10 \delta = 9 \delta$.

In contrast, EFS contracts achieve a much better revenue. Let the subsidies be $s_1 = 5\delta$ and $s_2 = 0$. The contract is set as $\alpha=\frac{3}{5}$. By assigning the task to agent $2$, we have EFS constraint for agent $1$ as
$5\delta = s_1 \ge (\alpha p_1 r  - c_1) +s_2 = 5\delta$,
and for agent $2$ as
$
\frac{1}{20}=(\alpha p_2  - c_2) +s_2  \ge s_1 = 5\delta$.
The principal gains a revenue of $\frac{1}{5}-5\delta$. The ratio $\frac{\frac{1}{5}-5\delta}{9\delta} \to \infty$ as $\delta\to 0$.
\end{example}

\section{Price of Fairness}

We now study the price of fairness (PoF), defined as the ratio between the optimal revenue without fairness
constraints and the optimal EFS revenue, to quantify the effect of EFS constraints on the principal's revenue:
\[\rm PoF =\frac{\opt}{\optefs}
\]
where $\opt$ is the optimal revenue without fairness constraints.
Although 
\cref{efs_over_ef} shows that EFS contracts could be arbitrarily better than EF contracts, it remains to understand how much revenue can be lost relative to the unconstrained optimum. \cref{main_theorem_pof} completely characterizes the price of fairness for EFS contracts, in contrast to the unbounded PoF in EF contracts. 
\begin{theorem}\label{main_theorem_pof}
    The Price of Fairness of EFS contracts is ${\rm PoF} = n^{n+O(1)}$.
\end{theorem}
Before proving these bounds, we establish a useful characterization of optimal single-task EFS contracts. 

\subsection{A Characterization for Optimal Single-task EFS Contracts}

We have the following key lemma for characterizing the optimal solution when assigning a single task in EFS contracts. Let the contract for this task be $\alpha \in [0, 1]$, and the reward be $r$. The utility agent $i$ is $U_i(\alpha) = \alpha p_i r - c_i$, if it is assigned this task and obligated to exert efforts.

\begin{lemma}\label[lemma]{lemma:min_sub}
Consider the setting of assigning a single task in EFS contracts. Let $\alpha \in [0, 1]$ be a contract such that there is at least one agent $i$ with $U_i(\alpha)\ge 0$. Under such a contract $\alpha$, the optimal solution must assign the task to the agent $i^*$ such that $U_{i^*}(\alpha) \ge 0$ and that $U_{i^*} (\alpha) \ge U_{j} (\alpha)$ for all agents $j$. Moreover, the subsidies that maximize the principal's revenue and are EFS are  $s_{i^*} = 0$ and $s_j = \max_{j' \neq {i^*}} \max\{U_{j'}(\alpha), 0\}$ for all agent $j\neq i^*$.
\end{lemma}

\begin{proof}
First, we note that a contract $\alpha$ such that there is at least one agent $i$ with $U_i(\alpha)\ge 0$ must exist, given \cref{asp:usefulness}. 
Suppose that the task is assigned to any agent $i$ with contract $\alpha$. Then, we have that the following EFS constraints must be satisfied:
\begin{equation}\label{singletaskefs}
\begin{aligned}
s_j \ge \max\{U_j(\alpha), 0\} + s_i,\quad  &\forall j \neq i\\
U_{i}(\alpha) +s_i \ge s_j,\quad  &\forall j \neq i\\
s_j \ge s_{j'},\quad  &\forall j \neq j'
\end{aligned}
\end{equation}
The first two constraints imply that $U_i (\alpha) \ge \max\{U_j(\alpha), 0\}$ must hold for all the agent $j\neq i$.
The third constraint implies that the subsidies for all agents $j\neq i$ must be the same.

    Let set $s_i=0$ and $s_j=\max_{j' \neq i} \max\{U_{j'}(\alpha), 0\}$ for all agent $j\neq i$. Clearly, these subsidies terms maximize the principal's revenue. Finally, 
    We only need to show that these are the smallest subsidies that guarantee the EFS constraints. Regarding agent $i$ this is trivial.
    Then, since $s_j=s_{j'}$,  we get $s_j\ge\max_{j' \neq i} \max\{U_{j'}(\alpha), 0\}$ for all $j\neq i$ for the first constraint. This concludes the proof.
\end{proof}

\subsection{Proof of \cref{main_theorem_pof}}

In particular, we show that for the price of fairness, the upper bound is $O((n+1)^n)$ (\cref{price_of_fair_EFS}) and the lower bound is $\Omega(n^{n-1})$ (\cref{priceoffainesslowerbound}).  These bounds differ only by a polynomial factor $O(n)$, together implying a PoF of $n^{n+O(1)}$.

\begin{figure}[t]
    \centering
    \includegraphics[width=0.7\textwidth]{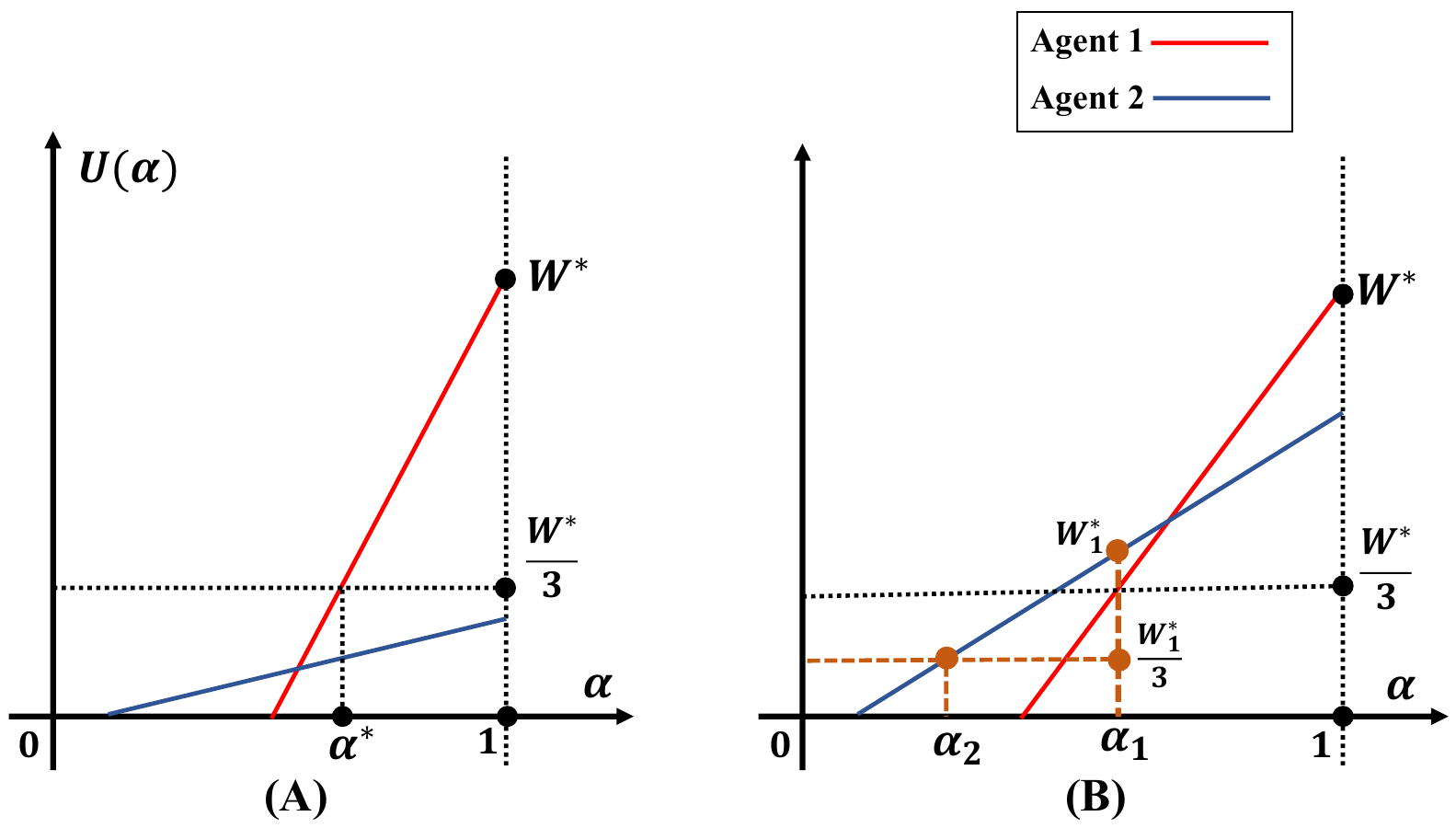}
    \caption{Illustration for $n=2$. The quantity $W^*$ is the largest revenue without fairness, obtained by assigning task to agent $1$. Figure (A) depicts an example where at contract $\alpha=\alpha^*$ satisfying $U_1(\alpha^*) = \frac{W^*}{3}$ and agent $1$ has the largest utility; then the contract $\alpha^*$ gives a $\frac{1}{3}$ approximation and the task is assigned to agent $1$. Otherwise, as in Figure (B), agent $2$ gives the largest utility. we repeat the process and compute a new contract $\alpha_2$ of agent $2$ satisfying $U_2(\alpha_2) = \frac{W^*_1}{3}$. In this case, the principal gains a revenue at least $\frac{W^*_1}{3}$, which is at least a $(\frac{1}{n+1})^n = \frac{1}{9}$ fraction of $W^*$. }
    \label{fig_illustaraint}
\end{figure}

\begin{proposition}\label[proposition]{price_of_fair_EFS}
    The price of fairness in EFS contracts is upper-bounded by $O((n+1)^n)$, where $n$ is the number of agents.
\end{proposition}

The key step for proving the upper bound in \cref{price_of_fair_EFS} is to show that $O((n+1)^n)$ is an upper bound for the case of designing an EFS contract for a single task. After showing this, by the fact that the principal's revenue is additive, we can extend the upper bound $O((n+1)^n)$ to the case of multiple tasks.

Consider the case of allocating a single task with reward $r$. For a contract $\alpha$, let $U_i(\alpha) \triangleq  \alpha p_i r - c_i$ be agent $i$'s (possibly negative) utility from exerting effort on the task.
Suppose that a contract $\alpha$ is given. \cref{lemma:min_sub} implies that in the optimal solution that maximizes the principal's revenue, this task must be assigned to an agent (e.g.,  denote it as agent $i$) that has the maximum nonnegative utility $\max\{U_i(\alpha),0\}$, and the subsidy terms are set accordingly.
Our proof mainly utilizes this observation to find an EFS contract, and the algorithm proceeds iteratively: in each iteration, where a contract is given and some agent is selected, it verifies whether the selected agent gives the largest nonnegative utility so that \cref{lemma:min_sub} can be applied.
To be more specific,  let $W_i$ be the welfare of agent $i$ (i.e., the principal's revenue from agent $i$ without fairness constraints), and let $i^*$ be the agent giving the largest welfare. In each iteration, compute the contract $\alpha^*$ such that $U_{i^*}(\alpha^*) = \frac{W_{i^*}}{n+1}$. If agent $i^*$ has the largest utility under contract $\alpha^*$, then we assign this task to agent $i^*$ and set the subsidies as in \cref{lemma:min_sub}. This clearly gives the principal at least revenue $\frac{W_{i^*}}{n+1}$. 
If agent $i^*$ does not have the largest utility, we move to the next iteration and obtain a new subproblem by redefining the welfare of agents as the agents' utilities (i.e., $\max\{U_i(\alpha^*), 0\}$) under contract $\alpha^*$. This process repeats at most $n$ times, which finally implies an approximation ratio of $(n+1)^n$ by induction. The high-level idea of the algorithm is illustrated in Figure~\ref{fig_illustaraint}.

\begin{algorithm}[t]
\caption{Computing the upper bound $O((n+1)^n)$ for EFS contract\\
\textbf{Input:}  $r$, $\{p_{i}\}_{i \in \mathcal{N}}$, $\{c_{i}\}_{i \in \mathcal{N}}$, $ n, \{W_i\}_{i \in \mathcal{N}}$ \\
\textbf{Output:} The contract ${\alpha}^*$ and the assigned agent $i^*$
}\label{algorithm_appr_EFS}
$W^0_i = W_i, \forall i \in [n]$ \;
\For {$k=1,2,\dots, n$}{\label{line2:appr_EFS}
$i^k = \arg\max_{i \in \mathcal{N}} W_i^{k-1}$\;
$T^k = \frac{W^{k-1}_{i^k}}{n+1}$, $\alpha^k = \frac{c_{i^k} + T^k}{p_{i^k} r}$ \label{line4:appr_EFS}\;
$i^* = \arg\max_{i \in \mathcal{N}} U_{i} (\alpha^k)$, {\text{tie breaks in favor of $i^k$}}\;
\If{$i^* = i^k$}{ \label{line6:appr_EFS}
$\alpha^* = \alpha^k$\;
break\;
}
$W_i^k = U_i(\alpha^k), ~\forall i \in \mathcal{N}$

}
\Return $\alpha^*$ and $i^*$.
\end{algorithm}

\begin{proof}
    We first prove the upper bound $O((n+1)^n)$. This starts by
showing that $(n+1)^n$ upper bound holds for the case with a single task, i.e., $m=1$.
Let the contract for this task be $\alpha \in [0, 1]$, and the reward be $r$. The utility of agent $i$ is $U_i(\alpha) = \alpha p_i r - c_i$, if it is assigned this task and exerts efforts. The welfare (i.e., the maximum principal's revenue) from assigning task to agent $i$ is $W_i = p_i r - c_i$.

Suppose that contract $\alpha$ is given such that there is at least one agent with nonnegative utility, i.e., $\max_{i}U_{i}(\alpha)\ge 0$. We know that such a contract must exist due to  \cref{asp:usefulness}. By \cref{lemma:min_sub}, the task must be assigned to the agent (e.g., agent $i$) with the largest utility, and the optimal subsidy terms are $s_i = 0$ and $s_j = \max_{j' \neq i} \max\{U_{j'}(\alpha), 0\}$ for all agent $j\neq i$. 
This motivates us to build a feasible EFS contract for a single task, and our algorithm is presented in Algorithm \ref{algorithm_appr_EFS}. 


We next show that Algorithm \ref{algorithm_appr_EFS} computes a contract that gives an $(n+1)^n$-approximation ratio for the single task case.
Let $W^* =\opt= \max_{i\in \mathcal{N}} W_i$ be the maximum revenue (i.e., the maximum welfare) without fairness constraints. The following lemma shows that the returned contract provides the principal's revenue at least $(\frac{1}{n+1})^n W^*$ while being EFS. .

\begin{lemma}\label[lemma]{lemma_of_single_appr}
    Algorithm \ref{algorithm_appr_EFS} returns a contract for the single task case, by which the principal gets revenue at least $(\frac{1}{n+1})^n W^*$.
\end{lemma}
\begin{proof}
First, notice that by setting $\alpha^k$ as in Line \ref{line4:appr_EFS}, we have that the agent $i^k$'s utility as $U_{i^k}(\alpha^k) = \alpha^k p_{i^k} r - c_{i^k} = T^k = \frac{W^{k-1}_{i^k}}{n+1}$, implying $\alpha^k\in [0, 1]$. Therefore, if $i^* \neq i^k$, we have that
\begin{equation}\label{eqation:n_1lowerbound}
W^{k}_{i^{k+1}}=\max_{i\in \mathcal{N}} W^k_i   > U_{i^k}(\alpha^k) = \frac{W^{k-1}_{i^k}}{n+1}
\end{equation}

On the other hand, if we have $i^* = i^k$ at Line \ref{line6:appr_EFS}, then it holds $U_{i^k}(\alpha^k) \ge U_j (\alpha^k) $ for all agent $j \neq i^k$. Clearly, $U_{i^k}(\alpha^k) \ge 0$ must hold by the definition of $W^k_i$. Therefore, by \cref{lemma:min_sub}, we have a feasible contracts with $s_{i^k} = 0$ and $s_{j} \le U_{i^k}(\alpha^k) = \frac{W^{k-1}_{i^k}}{n+1}$.
Moreover, note that $W_{i^k} \ge W^{k-1}_{i^k}$. This implies that the principal gains revenue from contract $\alpha^k$ by assigning task to agent $i^k$ is at least $W_{i^k} - \frac{n}{n+1}W^{k-1}_{i^k} \ge \frac{1}{n+1}W^{k-1}_{i^k} \ge (\frac{1}{n+1})^2 W^{k-2}_{i^{k-1}} \ge \dots \ge (\frac{1}{n+1})^k W^0_{i^1} = (\frac{1}{n+1})^k W^*$, where the inequality is by (\ref{eqation:n_1lowerbound}).

Further notice that $U_i(\alpha)$ is a linear function. Hence, by iterating at most $n$ times at Line \ref{line2:appr_EFS}, one must be able to find a $i^k$ such that Line \ref{line6:appr_EFS} is true. Indeed, if the algorithm does not stop, the next selected agent must have a strictly smaller slope $p_i r$ than the current selected agent. Therefore, the same agent cannot be selected again, and the algorithm stops after at most $n$ iteration.  Therefore, the principal obtains a solution by Algorithm \ref{algorithm_appr_EFS} with revenue at least $(\frac{1}{n+1})^n W^*$, whose contract is given in \cref{lemma:min_sub}.
\end{proof}

Next, we extend the analysis to the cases of multiple tasks, $m>1$. We apply Algorithm \ref{algorithm_appr_EFS} to each task $j\in \mathcal{M}$ separately. Finally, we get an allocation of tasks $\{S_i\}_{i \in \mathcal{N}}$. Let $s_{i,j}$ be the subsidy that agent $i$ obtains from task $j$ by Algorithm \ref{algorithm_appr_EFS}. Let $U_{i}(\alpha, j) = \alpha_j p_{i,j} r_j -c_{i, j}$.

Recall that for each task, Algorithm \ref{algorithm_appr_EFS} gives a fair contract as in \cref{lemma:min_sub}. This implies that for each $j\in S_i$, the effort constraint holds, i.e., $U_{i}(\alpha, j) \ge 0$.
Furthermore, we prove that the EFS constraints hold. That is, for any $i, j\in \mathcal{N}$, it holds that
\[
U_i(\alpha, S_i) + s_i \ge \sum_{k \in S_j} \max\{U_i(\alpha_k, k), 0\} + s_j
\]
where $s_i = \sum_{j\in \mathcal{M}} s_{i,j}$ and $U_i(\alpha, S_i) = \sum_{j\in S_i} U_{i}(\alpha_j, j)$.
That is equivalent to showing
\begin{small}
\begin{align*}
&\sum_{t \in S_i} U_{i} (\alpha_t, t) + \sum_{t \in S_i} s_{i,t} + \sum_{t \in S_j} s_{i,t} + \sum_{t \in \mathcal{M} \setminus \{S_i \cup S_j\}} s_{i,t} \\
&\ge \sum_{t \in S_j} \max\{ U_{i} (\alpha_t, t), 0\} + \sum_{t \in S_i} s_{j,t} + \sum_{t \in S_j} s_{j,t} + \sum_{t \in \mathcal{M} \setminus \{S_i \cup S_j\}} s_{j,t}
\end{align*}
\end{small}
Note that for $t \in \mathcal{M} \setminus \{S_i \cup S_j\}$, we have $\sum_{t \in \mathcal{M} \setminus \{S_i \cup S_j\}} s_{i,t} =\sum_{t \in \mathcal{M} \setminus \{S_i \cup S_j\}} s_{j,t} $ by \cref{lemma:min_sub}. Hence, EFS constraint becomes
\begin{align*}
&\sum_{t \in S_i} U_{i} (\alpha_t, t) + \sum_{t \in S_i} s_{i,t} + \sum_{t \in S_j} s_{i,t}  \\
&\quad \ge \sum_{t \in S_j} \max\{ U_{i} (\alpha_t, t), 0\} + \sum_{t \in S_i} s_{j,t} + \sum_{t \in S_j} s_{j,t}
\end{align*}
This inequality holds because
\[
\sum_{t \in S_i} U_{i} (\alpha_t, t) + \sum_{t \in S_i} s_{i,t} \ge \sum_{t \in S_i} s_{j,t}
\]
and
\[
\sum_{t \in S_j} s_{i,t}  \ge \sum_{t \in S_j} \max\{ U_{i} (\alpha_t, t), 0\} + \sum_{t \in S_j} s_{j,t}
\]
where the two inequailities follow the EFS constraint for the single task case, i.e., (\ref{singletaskefs}).

Therefore, by applying Algorithm \ref{algorithm_appr_EFS} to each task separately, we get EFS contracts for cases of multiple tasks. Note that for each task, the principal achieves a fraction of $(\frac{1}{n+1})^n$ of maximum welfare. Since the principal's objective is additive, the approximation ratio $({n+1})^n$ holds for the cases of multiple tasks.
\end{proof}

Next, we show that the price of fairness has a lower bound of $\Omega(n^{n-1})$. The instance construction is nontrivial: we build a single-task instance in which the agents’ success probabilities and costs take exponential-type forms. The proof is inspired by our algorithm for establishing an $O((n+1)^n)$ approximation in the single-task setting (\cref{lemma_of_single_appr}).

\begin{proposition}\label[proposition]{priceoffainesslowerbound}
    The price of fairness of EFS contracts is lower bounded by $\Omega(n^{n-1})$.
\end{proposition}
\begin{proof}
    We construct an instance of a single task that gives the above lower bound. Let the reward of the task be $r=1$, and there are $n\ge 3$ agents. To facilitate exposition, we define two parameters for our construction: $\beta=1+\frac{1}{n}$
and let
\[
\gamma = \frac{1}{2}\Big( \frac{n(\beta-1)}{\beta(n-1)}\Big)^{n-1}.
\]

For each agent $i$, let its success probability be
\[
p_i = \Big( \frac{\beta-1}{n-1} \Big)^{i-1}.
\]
Moreover, we let the welfare of agent $i$ exerting effort on the task be
\[
W_i = \gamma\Big( \frac{\beta}{n} \Big)^{i-1}
\]
Hence, by definition, the cost of agent $i$ should be $c_i = p_i - W_i$. First, we show that such a construction is feasible.
\begin{claim}\label[claim]{claim_feasible_piw}
    The instance constructed above is feasible, i.e., $0\le p_i, c_i, W_i \le 1$.
\end{claim}
\begin{proof}
    First, $p_i, W_i \ge 0$ is by definition. Moreover, we have $p_i = \Big( \frac{\beta-1}{n-1} \Big)^{i-1} = \Big( \frac{\frac{1}{n}}{n-1} \Big)^{i-1} = \Big( \frac{1}{n(n-1)} \Big)^{i-1} \le 1$, which is by $\frac{1}{n(n-1)} \le 1$. $W_i \le 1$ holds by that
    \begin{align*}
    \frac{W_i}{p_i} &= \gamma\Big( \frac{\beta(n-1)}{n(\beta-1)} \Big)^{i-1} =\frac{1}{2}\Big( \frac{n(\beta-1)}{\beta(n-1)}\Big)^{n-1} \Big( \frac{\beta(n-1)}{n(\beta-1)} \Big)^{i-1}\\
    &= \frac{1}{2} \Big( \frac{n(\beta-1)}{\beta(n-1)}\Big)^{n-i} \le \frac{1}{2}
    \end{align*}
    where the last inequality is by $\frac{n(\beta-1)}{\beta(n-1)} = \frac{1}{(1+\frac{1}{n})(n-1)} \le 1$. By the above argument, we also have $W_i \le p_i$. Hence, $0\le c_i \le 1$.
\end{proof}

Next, we show the optimal revenue without fairness, $\opt$. Recall that in this case, the principal can extract the full welfare as its own revenue. 

\begin{claim}\label[claim]{claim_without_fair_opt}
    The optimal revenue without fairness $\opt$ is obtained by assigning the task to agent $1$, i.e., $W_1$.
\end{claim}
\begin{proof}
    Note that $\opt$ is the full welfare that the principal can extract. Observe that $W_{i+1} = \Big( \frac{\beta}{n} \Big)W_i < W_i$ since $\frac{\beta}{n} <1$. Hence, the maximum welfare is obtained at $i=1$.
\end{proof}

In the next step, we show that the optimal revenue $\optefs$ obtained from EFS contracts is $W_n$, which assigns the task to agent $n$.

\begin{claim}
    The optimal revenue $\optefs$ obtained from EFS contracts is $W_n$
\end{claim}
\begin{proof}
    Let the utility of agent $i$ is $U_i(\alpha) = \alpha p_i - c_i$, which is a linear function in $\alpha$ with slope $p_i$.
    Let $\alpha_i$ be the contract where agent $i$ and agent $i+1$ have the same utility, i.e., $U_i(\alpha_i) = U_{i+1}(\alpha_i)$. By our construction, we have  
    \[
    d_i := (1-\alpha_i) = \frac{W_i - W_{i+1}}{p_i - p_{i+1}} = \frac{W_i (1-\frac{\beta}{n})}{p_i (1-\frac{\beta-1}{n-1})} = \frac{W_i}{p_i}\frac{n-1}{n}
    \]
    At the contract $\alpha_i$, we then have 
    \begin{equation}\label{utiltywiovern}
    U_i(\alpha_i) = U_{i+1}(\alpha_i) = W_i - (1-\alpha_i)p_i = \frac{W_i}{n}
    \end{equation}
    Moreover, observe that $d_{i+1}>d_i$, which is by 
    \[
    \frac{d_{i+1}}{d_i} = \frac{\frac{W_{i+1}}{p_{i+1}}\frac{n-1}{n}}{\frac{W_i}{p_i}\frac{n-1}{n}} = \frac{\beta/n}{(\beta-1)/(n-1)} =\frac{n^2-1}{n} >1
    \]
    This implies the order for contracts $\{\alpha_i\}$ is $\alpha_1 > \alpha_2 >\dots > \alpha_{n-1} $. Furthermore, observe that the order for slopes of $U_i(\alpha)$ is $p_1 > p_2 >\dots>p_{n-1}$. Let $\alpha_0 = 1$. Both observations together imply that for any contract $\alpha$ in the interval $\alpha \in [\alpha_i, \alpha_{i-1})$, agent $i$ has the largest utility, and the optimal EFS contract is to assign the task to agent $i$ by \cref{lemma:min_sub}. 

Now, consider contract $\alpha \in [\alpha_i, \alpha_{i-1})$. \cref{lemma:min_sub} implies that agent $j\neq i$ receives the same subsidy $s_j$. To maximize the principal's revenue, it is clearly optimal to set the contract as $\alpha=\alpha_i$. Therefore, by (\ref{utiltywiovern}), we have that $s_j = \frac{W_i}{n}$. Therefore, the principal's optimal revenue for the interval $[\alpha_i, \alpha_{i-1})$ is at contract $\alpha=\alpha_i$
\[
W_i - U_{i}(\alpha_i) - (n-1)s_j = W_i - \frac{W_i}{n} - \frac{W_i(n-1)}{n} = 0
\]
Hence, by considering assigning the task to any agent $i < n$, we have the conclusion that the principal's optimal revenue is $0$.

Finally, we consider agent $n$. If assigning the task to agent $n$, the effort constraint implies that 
\[
\alpha p_n -c_n \ge 0,
\]
and hence, we have $\alpha \ge \alpha_n =  \frac{c_n}{p_n} = 1-\frac{W_n}{p_n}$. By the same arguments as above, we know that for any contract in the interval $\alpha \in [\alpha_n, \alpha_{n-1})$, the optimal EFS contract would assign the task to agent $n$. Therefore, to maximize the principal's revenue, we set the contract as $\alpha=\alpha_n$. Under this contract, all other agents $j \neq n$ would receive non-positive utility from the task, i.e., 
\begin{align*}
    U_j(\alpha_n) &= \alpha_n p_j - c_j \\
    &= (1-\frac{W_n}{p_n})p_j - (p_j - W_j) = W_j - W_n \frac{p_j}{p_n} \le 0
\end{align*}
where the last inequality is by that $\frac{W_i}{p_i} =  \gamma\Big( \frac{\beta(n-1)}{n(\beta-1)} \Big)^{i-1} \le \gamma\Big( \frac{\beta(n-1)}{n(\beta-1)} \Big)^{i} = \frac{W_{i+1}}{p_{i+1}}$ since $\frac{n(\beta-1)}{\beta(n-1)} = \frac{1}{(1+\frac{1}{n})(n-1)} \le 1$. Therefore, the optimal subsidies for all agents $j\neq i$ are $s_j = 0$. \cref{lemma:min_sub} also implies that $s_n=0$. Moreover, the agent $n$ receives {\it zero} utility at contract $\alpha_n$. Hence, the optimal principal's revenue is $W_n$, concluding the proof.
\end{proof}
Now, we have the price of fairness for the constructed instance as 
\[
\frac{\opt}{\optefs} = \frac{W_1}{W_n} = \frac{\gamma}{\gamma(\beta/n)^{n-1}} = \frac{n^{n-1}}{(1+\frac{1}{n})^{n-1}} \ge \Omega(n^{n-1})
\]
where the last inequality holds due to $(1+\frac{1}{n})^{n-1}\le (1+\frac{1}{n})^n \le e$. This concludes the proof.
\end{proof}

\section{Computational Results}

In this section, we present computational results of EFS constructs. Compared with EF contracts, EFS contracts have $n$ additional subsidy variables. Because the two formulations are closely related, a natural first question is how EF and EFS contracts are connected. Our first result shows that an
EFS instance can be reduced to an EF instance in polynomial time. Missing proofs are in the appendix.

\begin{proposition}\label[proposition]{reudtionc_propo}
    Given an EFS contract instance, we can reduce it to the EF contract problem in polynomial time.
\end{proposition}

\citet{castiglioni2025faircontracts} shows that it is \NP-hard to compute the optimal EF contracts, whereas the reduction in \cref{reudtionc_propo} shows that EFS contracts are no harder than EF contracts from
an oracle-reduction perspective. One might therefore hope that subsidies make the optimization problem computationally easier. This hope may be further supported by the fixed-contract case: when the task-level contracts are fixed, computing an EF allocation (without subsidies) is closely related to fair division of indivisible goods and can be \NP-hard, whereas computing envy-free allocations with subsidies becomes relatively easy, as shown in~\citep{halpern2019fair}.

However, this intuition fails once the principal must jointly optimize over task allocations, agents' incentives, and subsidies. Although subsidies make it easier to restore exact fairness for a fixed contract, they do not eliminate the computational difficulty of finding the revenue-maximizing EFS contract. 
\cref{them_np} shows that computing optimal EFS contracts is also \NP-hard in general. 

\begin{theorem}\label{them_np}
    Unless $P=\NP$, there is no polynomial-time algorithm for computing the optimal EFS contracts.
\end{theorem}

We prove \cref{them_np} by reducing from the $3$-Partition problem. Intuitively, the construction is to establish a connection to EF contracts: an optimal EFS contract in the constructed instance becomes an EF contract when a partition exists. Our construction consists of one ``efficient" agent, which contributes a large revenue, and $q$ ``normal" agents, which contribute relatively small revenue. We show that if a $3$-partition exists, the EFS contracts achieve at least a revenue $\frac{1}{2}$. In contrast, if there does not exist a $3$-partition, the revenue of EFS contracts is strictly less than $\frac{1}{2}$. The proof of the latter case is more involved. We prove it by contradiction and show that if the revenue of EFS contracts is at least  $\frac{1}{2}$, the sets of tasks assigned to $q$ normal agents will form a valid $3$-partition solution.

While computing the optimal EFS contracts is \NP-hard in general, we show that when the number of tasks is a constant, there exists a polynomial-time algorithm for computing the optimal EFS contract.

\begin{proposition}\label[proposition]{proposition_polytime_constant}
    When the number of tasks is a constant, there exists a polynomial-time algorithm for computing the optimal EFS contract.
\end{proposition}

\section{Conclusions}
We study envy-free contracts with subsidies as a way to address the fairness-revenue dilemma raised in prior works on fair contract design. In particular, EFS does not relax envy-freeness; it restores exact envy-freeness through compensatory transfers. Our results show that the resulting revenue loss is no longer unbounded, but is tightly characterized by $n^{n+O(1)}$.

One interesting open question left for the EFS contracts problem is whether there exists a polynomial-time algorithm when the number of agents is constant. We leave this question for future work.

\newpage
\bibliographystyle{apalike}
\bibliography{ref}
\newpage



\newpage
\appendix

\section{Omitted Proofs}

\subsection{Proof of \cref{propo_existend_efs}}
\begin{proof}
\citet{castiglioni2025faircontracts} shows that EF contracts always exist under \cref{asp:usefulness}, which implies that EFS contracts always exist as well. Therefore, we can find one EFS construct by setting all subsidies as $s_i =0$ and assigning each task $j$ to the agent with the smallest wage $\min_{i} \frac{c_{i,j}}{p_{i,j}r_j}$ (where we let $\frac{0}{0}=0$).
\end{proof}

\subsection{Proof of \cref{reudtionc_propo}}
\begin{proof}
Consider \eqref{efs_program}.
Let $\opt^1$ be the optimal revenue of EFS contracts. Given an EFS contracts instance, we construct the following EF contract instance: Consider a second instance with $2m+n$ tasks, where we add $m+n$ tasks (denoted as set $\mathcal{M^+}$) to EFS contract instance. The success probabilities and costs for any agent $i\in \mathcal{N}$ and task $k\in \mathcal{M}^+$ are $p_{i,k}=1$ and $c_{i,k}=0$. The reward for task $k\in \mathcal{M}^+$ is $r_k=1$.

For each agent $i$, we denote $S^1_i$ as the set of tasks assigned to agent $i$ from set $\mathcal{M}$, and as $S^2_i$ the tasks assigned to agent $i$ from $\mathcal{M}^+$. Hence, the program for the constructed EF contract is
\begin{align*}
 \max \quad & \sum_i[\sum_{k \in S^1_i} (1-\alpha_k) p_{i, k }r_k + \sum_{k \in S^2_i} (1-\alpha_k) p_{i, k }r_k] \\
 \text{s.t.} \quad    & \sum_{k \in S^1_i} \alpha_k p_{i, k }r_k - c_{i,k} + \sum_{k \in S^2_i} \alpha_k p_{i, k }r_k \\
 &\quad \geq \sum_{k \in S^1_j} \max\{\alpha_k p_{i, k} r_k - c_{i, k},0\} + \sum_{k \in S^2_j} \alpha_k p_{i, k }r_k,  \\
 &\qquad \qquad \qquad \qquad \qquad \qquad \qquad \qquad\quad \forall j \neq i \\
 &\text{(\ref{effort-constraints})}
\end{align*}
Let $\opt^2$ be its optimal value. Next, we show that solving the constructed EF contract instance optimally is equivalent to solving the EFS contract optimally. Specifically, we have
\begin{claim}\label{claim56opt}
    $\opt^1=\opt^2-m-n$.
\end{claim}
\begin{proof}
    Given an optimal solution to the EFS program $S_i,s_i,\alpha_j$ for all $i \in \{\mathcal{N}\}$ and $j \in \{\mathcal{M}\}$,
we build a solution to the EF program by setting $S^1_i=S_i$ and set the same contract $\alpha_j$ for tasks $j\in \mathcal{M}$.

First, we assign $\lceil s_i\rceil$ tasks from $\mathcal{M}^+$ to agent $i$, $S^2_i$. Additionally, we set the contract $\alpha_j=1$ for $\lfloor s_i\rfloor$ tasks in $S^2_i$. If there remains (at most) one task contract unset, set its contract as    $s_i-\lfloor s_i\rfloor$. By this, we have that $\sum_{k \in S^2_i} \alpha_k p_{i, k }r_k = s_i$. We apply this operation to all agents $i \in \mathcal{N}$. If there still exist some tasks in $\mathcal{M}^+$ not assigned,  assign them to an arbitrary agent with contract $0$. Note that the number of total tasks assigned in $\{S_i^2\}$ is less than $|\mathcal{M}^+|$, since in any optimal EFS solution we must have $\sum_{i} s_i \le m$ implying $\sum_{i} \lceil s_i\rceil \le \sum_i (s_i + 1) \le m + n$.

    It is easy to see that the contract and allocation is EF, and the principal's revenue in the constructed instance is $\opt^1+m+n$.

    Conversely, given an optimal solution to the constructed EF  program, we can build a solution to the EFS program by taking the same allocation $S_i = S_i^1$  and setting the same contracts $\alpha_j$ for $j\in \mathcal{M}$. Additionally, set the subsidy $s_i= \sum_{k \in S^2_i} \alpha_k p_{i, k }r_k$ for agent $i$. It is easy to see that EFS constraints are satisfied, while the principal's revenue decreases by $m+n$ (as we remove the added tasks).
    Hence, we get $\opt^2-m-n$.
    This concludes the proof.
\end{proof}

Hence, by Claim \ref{claim56opt}, we have that if one can solve EF contracts optimally, then one can solve EFS optimally. This proves the proposition.
\end{proof}

\subsection{Proof of \cref{them_np}}

\begin{proof}
    We prove the \NP-hardness by reducing from the $3$-Partition problem~\citep{garey1979computers}: Given $3q$ positive integers $a_1, a_2, \dots, a_{3q}$, and a integer $B>0$, such that $\sum_{i=1}^{3q} a_i = qB$, the problem is to determine whether these integers can be partitioned into $q$ triplets $S_1, \dots,S_q$ such that for each triplet $\sum_{k \in S_i} a_k = B$ holds. The 3-partition problem remains strongly NP-complete under the restriction that every integer in S is strictly between $\frac{B}{4}$ and $\frac{B}{2}$.

    Given an instance of $3$-Partition problem with every integer $\frac{B}{4}<a_i <\frac{B}{2}$, we construct an instance of EFS contract problem: Let $h = \frac{1}{4(q+2)}$. Define $q+1$ agents, denoted as $\mathcal{N} = \{0, 1, \dots, q\}$. There are $3q+1$ tasks. For each integer $a_k$, create an item task $t_k$. Additionally, we create a special task $g$. Denote the set of tasks as $\mathcal{M}$. Let the reward of each task be $r_j = 1$, all $j\in \mathcal{M}$. Next, we define the probabilites and costs: 
    \begin{itemize}
        \item Consider the item task $t_k$. For each agent $i \in \{1, \dots, q\}$, define the probability $p_{i, t_k} = h\frac{a_k}{B} \le 1$ and cost  $c_{i, t_k} = 0$. For agent $0$, define $p_{0, t_k} = 0$ and cost  $c_{0, t_k} = 0$.
        \item Consider the special task $g$. For each agent $i \in \{1, \dots, q\}$, define the probability $p_{i, g} = 2h$ and cost  $c_{i, g} = 0$. For agent $0$, define $p_{0, g} = 1$ and cost  $c_{0, g} = \frac{1}{2}$.
    \end{itemize}

Our proof is mainly to show that the principal receives revenue at least $\frac{1}{2}$ if and only if there exists a $3$-partition to the instance. The key intuition is that in the optimal solution of the constructed instance, the subsidy terms are all {\it zero} for all the agents. Hence, the EFS contracts problem indeed is reduced to EF contracts problem, which is \NP-hard by \cite{castiglioni2025faircontracts}.  

\begin{lemma}\label[lemma]{lemma_np1}
    If there exists a $3$-partition solution to the given instance, the principal receives revenue at least $\frac{1}{2}$.
\end{lemma}

\begin{proof}
    Let the $3$-partition solution be $S_1, S_2, \dots, S_q$. For each triplet $S_i$, assign all the tasks $j\in S_i$ to agent $i \in \{1, 2, \dots, q\}$. Moreover, we set the contract for the task $j \in S_i$ as $\alpha_j = 1$. Moreover, assign the special task $g$ to agent $0$, and set the contract $\alpha_g = \frac{1}{2}$. We set all the subsidies term $S_i = 0$ for all agents $i \in \mathcal{N}$.

    Now, we verify that this constructed contracts and allocation satisfy EFS constraints. Consider any agent $i, j \neq 0$. The utility that agent $i$ receives utility from its assigned set $S_i$ is
    \[
    \sum_{k \in S_i} \alpha_k p_{i,k} - c_{i,k} = h\frac{\sum_{k \in S_i} a_k}{B} = h
    \]
    Moreover, the utility that agent $i$ can obtains from agent $j$'s assigned set $S_j$ is
    \[
    \sum_{k \in S_j} \alpha_k p_{i,k} - c_{i,k} = h\frac{\sum_{k \in S_j} a_k}{B} = h
    \]
    Hence, there is no envy among agents $i, j$. Note that the utility agent $i$ receives from agent $0$'s assigned task $g$ is $\alpha_g p_{i,g} - c_{i,g} = h$. Hence, agent $i$ does not envy agent $0$.

    Now, consider agent $0$. Its utility receives from task $g$ is $\alpha_g p_{0,g} - c_{0,g} = 0$. Moreover, since agent $0$ has sucess probabilites $0$ for all other tasks, it receives utility $0$ from all other agents' assigned sets of tasks. Hence, agent $0$ does not envy other agents.

Therefore, the revenue that the principal receives is mainly from the task $g$, which is $(1-\alpha_g)p_{0, g} = \frac{1}{2}$. This concludes the proof.
\end{proof}

Next, we present the proof for the more challenging opposite direction, where there does not exist a $3$-partition.

\begin{lemma}\label{lemma_np2}
    If there does not exist a $3$-partition solution to the given instance, the principal receives revenue strictly less than $\frac{1}{2}$.
\end{lemma}
\begin{proof}
    We prove this lemma by contradiction. Suppose that the revenue that the principal receives is at least $\frac{1}{2}$. In this case, we note that the task $g$ must be assigned to agent $0$. Otherwise, the revenue that the principal can receives is at most
    \begin{align*}
    \sum_{k \in \mathcal{M}} p_{i,k} &= \sum_{k \in [3q]} p_{i,k} + p_{i, g} \\
    &= 2h + h\frac{\sum_{k \in [3q]} a_k}{B} = qh+2h = \frac{1}{4} < \frac{1}{2}
    \end{align*}
    which is obtained by assigning all the tasks to an arbitrary agent $i\neq 0$. Hence, the task $g$ must be assigned to agent $0$.

    Let $\alpha_g$ be the contract for task $g$. By effort constraints of agent $0$, we know that we must have $\alpha_g \ge \frac{1}{2}$. Let the feasible solution to the EFS contracts be $S_0, S_1, \dots, S_q$, where $g \in S_0$. Moreover, let the subsidies term be $s_0, s_1, \dots, s_q$. Denote $U_{i}(S):= \sum_{k\in S} \max\{ \alpha_k p_{i,k} - c_{i,k}, 0 \}$ as agent $i$'s utility from exerting efforts on set $S$. 
    Hence, we have $U_i(S_i)$ be agent $i$'s utility of the feasible solution. Then, we have the sum of the utilities of all the agent $i\neq 0$ as
    \begin{equation}\label{upp_tot_ul}
    \begin{aligned}
    T:&=\sum_{i\in \mathcal{N}\setminus\{0\}} U_{i}(S_i) = \sum_{i\in \mathcal{N}\setminus\{0\}} \sum_{k \in S_i } \alpha_k p_{i,k} \\
    &=  \sum_{k \in \mathcal{M}\setminus S_0 } \alpha_k h\frac{a_k}{B} \le \sum_{k \in \mathcal{M}\setminus \{g\} } h\frac{a_k}{B} = qh
    \end{aligned}
   \end{equation}
    Now, consider agent $i\neq 0$ EFS constraints which compares with the set of tasks of agent $0$, we have 
    \begin{equation}\label{equtilityEFS}
    U_{i}(S_i) +s_i \ge U_{i}(S_0) +s_0 \ge \alpha_g p_{ig} +s_0 
    \end{equation}
    Hence, we have the total subsidies of all agents $i\neq 0$ by summing up (\ref{equtilityEFS}) as 
    \begin{equation}\label{lb_sub}
    \sum_{i \in \mathcal{N}\setminus\{0\}} s_i \ge 2q\alpha_g h+qs_0 -\sum_{i\in \mathcal{N}\setminus\{0\}} U_{i}(S_i) \ge \alpha_g 2qh -T
    \end{equation}
    Moreover, consider agent $0$'s EFS constraintscompares with the set of tasks of agent $i$, we have 
    \begin{equation}\label{eqs0s1}
    \sum_{j\in S_0} \alpha_j p_{0,j} -c_{0,j} + s_0= \alpha_g -\frac{1}{2} +s_0 \ge 0+s_i = s_i
    \end{equation}
    \begin{claim} \label[claim]{claim_g}
        The optimal EFS contract has $S_0=\{g\}$.
    \end{claim}
        \begin{proof}
        If there exists one item task assigned to agent $0$ in an optimal EFS contract, one can construct an EFS solution that has the weakly higher principal's revenue by assigning that task to an arbitrary agent $i\neq 0$ with contract $0$. This is due to that agent $0$ receives utility $0$ from any item task by construction. Hence, assigning it to any agent $i\neq 0$ does not cause agent $0$'s envy. Moreover, it does not cause envy among any agents $i, j\neq 0$, since the contract for that task is $0$, leading to {\it zero} utility to agent $i\neq 0$. Finally, it decreases the utility that agent $i$ obtains from agent $0$'s bundle, further strengthening the EFS constraints. Hence, the new solution is EFS. Finally, since the principal gains zero revenue from this task in the original but gains positive revenue in the modified solution, the principal's revenue strictly improves, contradicting the optimality.
    \end{proof}

    Hence, we have the principal's revenue as 
    \begin{align*}
    &\sum_{i \in \mathcal{N}} \sum_{j\in S_i} (1-\alpha_j) p_{i,j} - \sum_{i \in \mathcal{N}} s_i \\&\le (1-\alpha_g) + \sum_{j \in \mathcal{M}\setminus\{g\}} h\frac{a_j}{B} - T  - \sum_{i \in \mathcal{N}} s_i\\
    &\le (1-\alpha_g) + \sum_{j \in \mathcal{M}\setminus\{g\}} h\frac{a_j}{B} - T  - (\alpha_g q2h -T) \\
    & = 1-\alpha_g (1+2qh) + qh
    \end{align*}
    where the first inequality is by Claim~\ref{claim_g} and the fact that the principal's revenue from item tasks is the welfare minus the agents' utility from item tasks i.e., $T$, and the second inequality is by $s_0\ge 0$ and Equation (\ref{lb_sub}). 

    Since $\alpha_g\ge \frac{1}{2}$, we have that the principal's revenue is at most $1-\alpha_g (1+2qh) + qh\le \frac{1}{2}$. Together with our assumption that the principal receives revenue at least $\frac{1}{2}$. Hence, the principal's revenue must be $\frac{1}{2}$, and therefore, $\alpha_g=\frac{1}{2}$ must hold.

    By (\ref{equtilityEFS}), we have 
    $U_{i}(S_i) +s_i \ge h +s_0,$
    and by (\ref{eqs0s1}), we have $s_0\ge s_i$,
    which together implies that 
\begin{equation}\label{low_utility}
U_i(S_i) \ge h
\end{equation}
By summing up over all the agent $i\neq 0$, we have 
\[
\sum_{i} U_i(S_i) \ge qh
\]
Then, by (\ref{upp_tot_ul}), we have $\sum_{i} U_i(S_i) = qh$, implying that (\ref{low_utility}) holds with equality and $\alpha_k =1$ for all the tasks except the special task $g$.

Note that $U_i(S_i) = \sum_{k \in S_i} \alpha_k h \frac{a_k}{B}$. This implies for each agent $i\neq 0$, we have $\sum_{k\in S_i} a_k =B$. Since each integer $\frac{B}{4}<a_i<\frac{B}{2}$, it implies that every set $S_i$ contains exactly $3$ tasks, implying a $3$-partition. Hence, we reach a contradiction to the assumption, concluding the proof.
\end{proof}

The proof concludes by Lemma~\ref{lemma_np1} and Lemma~\ref{lemma_np2}.
\end{proof}

\subsection{Proof of \cref{proposition_polytime_constant}}
\begin{proof}
    Note that since the number of tasks is a constant, we know that the number of allocations of assigned sets of tasks $(S_1, S_2,\dots, S_n)$ is polynomial $O(n^m)$. Observing that when fixing the allocations $(S_1, S_2,\dots, S_n)$, the program $\eqref{efs_program}$ immediately reduces to linear program,
    \begin{align*}
     \max_{\alpha, s} \quad &  \rev(S,\alpha,s) \\
    \text{s.t.} \quad &  \sum_{k \in S_i} \alpha_k p_{i, k }r_k - c_{i,k} + s_i \\
    & \qquad \qquad \qquad  \geq \sum_{k \in S_j} z_{i,k} + s_j,  ~~ \forall j \neq i, i,j\in \mathcal{N} \\
    & z_{i,k} \ge \alpha_k p_{i, k} r_k - c_{i, k},~ z_{i,k}\ge 0, \quad \forall k\in S_j i,j \in\mathcal{N} \\
    &\alpha_j p_{i,j} r_j - c_{i,j}\ge 0, \quad \forall i \in \mathcal{N}, j\in S_j \\
     &s_i\ge 0, \quad  \forall i\in \mathcal{N} 
\end{align*}
where we introduce a new variable $z_{i,j} \ge 0$ to linearize $\max\{\alpha_k p_{i,k}r_k -c_{i,k}\}$ for the $\eqref{efsconstraints}$. Since linear programs can be solved in polynomial time and only polynomially many allocations must be enumerated, the complexity for computing optimal EFS contracts is polynomial.
\end{proof}

\end{document}